\documentclass[conference,letterpaper]{IEEEtran}
\IEEEoverridecommandlockouts

\usepackage[utf8]{inputenc} 
\usepackage[T1]{fontenc}
\usepackage{url}
\usepackage{ifthen}
\usepackage{cite}
\usepackage[cmex10]{amsmath} 

\usepackage{amsfonts}
\usepackage{amssymb}
\usepackage{amsthm}
\usepackage{enumitem}
\usepackage{xcolor}
\usepackage{float}
\usepackage{bbm}

\usepackage{pgfplots}
\pgfplotsset{compat=1.17}

\newtheorem{theorem}{Theorem}

\newtheorem{definition}{Definition}
\newtheorem{corollary}{Corollary}

\theoremstyle{definition}
\newtheorem{remark}{Remark}

\theoremstyle{definition} \newtheorem{example}{Example}

\def \extended {1}

\begin{document}
\title{Generalized Information Inequalities via Submodularity, and Two Combinatorial Problems} 

\author{
  \IEEEauthorblockN{Gunank Jakhar, Gowtham R. Kurri, Suryajith Chillara}
  \IEEEauthorblockA{International Institute of Information Technology, Hyderabad \\
                    Hyderabad, India.\\ 
                    Email: \{gunank.jakhar@research., gowtham.kurri@, suryajith.chillara@\}iiit.ac.in}
  \and
  \IEEEauthorblockN{Vinod M. Prabhakaran}
  \IEEEauthorblockA{Tata Institute of Fundamental Research\\
                    Mumbai, India.\\
                    Email: vinodmp@tifr.res.in}
}

\maketitle

\begin{abstract}
  It is well known that there is a strong connection between entropy inequalities and submodularity, since the entropy of a collection of random variables is a submodular function. Unifying frameworks for information inequalities arising from submodularity were developed by Madiman and Tetali (2010) and Sason (2022). Madiman and Tetali (2010) established strong and weak fractional inequalities that subsume classical results such as Han's inequality and Shearer's lemma. Sason (2022) introduced a convex-functional framework for generalizing Han's inequality, and derived unified inequalities for submodular and supermodular functions. In this work, we build on these frameworks and make three contributions. First, we establish convex-functional generalizations of the strong and weak Madiman and Tetali inequalities for submodular functions. Second, using a special case of the strong Madiman–Tetali inequality, we derive a new Loomis–Whitney–type projection inequality for finite point sets in $\mathbb{R}^d$, which improves upon the classical Loomis–Whitney bound by incorporating slice-level structural information. Finally, we study an extremal graph theory problem that recovers and extends the previously known results of Sason (2022) and Boucheron~\emph{et al.}, employing Shearer’s lemma in contrast to the use of Han’s inequality in those works.  
\end{abstract}

\section{Introduction}

A set function $f:2^{[1:n]}\rightarrow \mathbb{R}$ is said to be submodular if it satisfies the diminishing returns property, i.e., the additional value gained by adding an element to a set decreases as the set grows larger~\cite{Fujishige05}. Submodular functions have applications in machine learning~\cite{krause2007near}, combinatorial optimization~\cite{ConfortiC84}, algorithmic game theory~\cite{lehmann2001combinatorial}, social networks~\cite{kempe2003maximizing}, and statistical physics~\cite{VanEnterFS93}. Matroid rank function, cut function in graphs, log-determinant functions associated with positive semidefinite matrices, coverage functions are some examples of submodular functions. Another example is entropy: for any collection of random variables $X_1,X_2,\dots,X_n$, the function $f(S)=H(X_S)$ is submodular~\cite{Fujishige05}. This observation has played a central role in connecting combinatorial inequalities with information-theoretic ones, and has led to a rich body of work exploring inequalities that hold for general submodular functions and their implications for entropy.

When specialized to entropy, submodularity gives rise to a number of fundamental information-theoretic inequalities. Classical examples include Han's inequality~\cite{Han78}, Shearer's lemma~\cite{ChungGFS86,Radhakrishnan03}, and Fujishige's inequalities~\cite{Fujishige78}. All these inequalities relate the joint entropy of a collection of random variables to the entropies of their subsets. A unifying framework was developed by Madiman and Tetali~\cite{MadimanT10}, who established strong and weak forms of inequalities for submodular functions. Their results subsume Han's inequality, Shearer's lemma, and Fujishige's inequalities as special cases, and clarify the role of fractional partitions, packings, and coverings in this context. Necessary and sufficient conditions for equality to hold in the weak form of Madiman-Telali's inequalities were obtained in \cite{JakharKCP25}. Sason~\cite{Sason22} introduced a convex-functional approach to deriving families of inequalities for submodular functions, extending the scope of classical results. In this framework, convex functions are applied to normalized submodular function values, yielding new inequalities and refinements of existing ones. Several other works have also explored information inequalities exploring submodularity, including \cite{kishi2014entropy,tian2011inequalities,iyer2021generalized}.

In addition to information theory, entropy inequalities find applications in a wide range of areas such as combinatorics and graph theory~\cite{Radhakrishnan03,MadimanT10}, concentration inequalities~\cite{BoucheronLM13,GavinskyLSS15}, statistical inference~\cite[Chapter~8]{polyanskiy2025information}, etc. In combinatorics, Shearer's Lemma has been utilized to establish an upper bound on the size of intersecting families of subgraphs within a complete graph~\cite{ChungGFS86}. Friedgut and Kahn~\cite{FriedgutK98} applied it to bound the number of copies of a hypergraph contained within another hypergraph. Kahn~\cite{Kahn01} derived an upper bound on the number of independent sets in a regular bipartite graph, which was subsequently generalized to any arbitrary graph in \cite{MadimanT10} by employing the strong form of Madiman and Tetali's inequalities. More broadly, Friedgut~\cite{Friedgut04} showed that various forms of Shearer's lemma recover several classical inequalities, including the Cauchy-Schwarz inequality, H\"older's inequality, the monotonicity of the Euclidean norm, and the monotonicity of weighted means. 

Radhakrishnan~\cite{Radhakrishnan03} provides a comprehensive exposition on using the entropy method as a framework for counting problems and surveys many of the aforementioned applications. In particular, he gives an entropy-based proof of the classical Loomis-Whitney projection inequality~\cite{LoomisH49}. This inequality has been used for demonstrating lower bounds on communication costs in distributed-memory matrix multiplication~\cite{Irony04}. More recently, Sason~\cite{Sason22} studied an extremal graph theory problem using Han's inequality for families consisting of all sets of equal size, thereby generalizing an edge isoperimetric inequality on the hypercube previously considered in~\cite{BoucheronLM13}.  

In this work, we leverage submodularity to establish generalized information inequalities, and we investigate two combinatorial problems using existing information inequalities. Our main contributions are as follows:
\begin{itemize}[leftmargin=*]
    \item We first establish convex-functional generalizations of the strong and weak forms of the Madiman-Tetali's inequalities for submodular functions~\cite{MadimanT10}, in which a convex function $g$ is applied to normalized submodular function values associated with a fractional partition (Theorem~\ref{SasonImprov} and Corollary~\ref{WeakSasonImprov}). Special cases of these inequalities yield bounds that are stronger than the corresponding result in \cite[Equation (22)]{Sason22} as well as the entropy power inequalities for joint distributions~\cite[Corollary~VI]{MadimanT10}.
    \item We next apply a special case of the strong form of the Madiman-Tetali's inequalities~\cite{MadimanT10} to derive a new Loomis-Whitney-type projection inequality for point sets in $\mathbb{R}^d$ (Theorems~\ref{StrongLW} and \ref{StrongLW-ddim}). The resulting bound strengthens the classical Loomis-Whitney inequality~\cite{LoomisH49,Radhakrishnan03} by incorporating additional slice-level structural information. 
    \item Finally, we study an extremal graph theory problem motivated by confusability graphs of noisy channels, building on and extending the formulation considered in \cite{Sason22}. In contrast to the use of Han's inequality in \cite{Sason22} for families consisting of all sets of equal size, we employ Shearer's lemma, a special case of the weak form of Madiman-Tetali's inequalities that applies to arbitrary set families, to obtain upper bounds on the number of edges (Theorem~\ref{ExtremalGraph}) that recover and extend previously known results~\cite[Theorem~4.2]{BoucheronLM13} and \cite[Theorem~2]{Sason22}.  
\end{itemize}

\section{Preliminaries}
\textit{Notation}: We use $[i : i+k]$ to represent the set $\{i, i+1, \ldots, i+k-1, i+k\}$, where $i,k\in\mathbb{N}$. For any set $A \subseteq [1:n]$, $< A$ denotes the set of all indices in $[1:n]$ that are less than every index in $A$, and $> A$ denotes the set of all indices in $[1:n]$ that are greater than every index in $A$. The power set of a set $A$ is denoted by $2^A$. $A^{\text{c}}$ denotes the complement of a set $A$. We use $\mathcal{F}$ to denote a family of subsets of $[1:n]$ allowing for repetitions, represented as $\{\!\!\{\cdot\}\!\!\}$.

\begin{definition}[Sub/Supermodular, and Modular Functions~\cite{Fujishige05}]
A set function $f: 2^{[1:n]} \rightarrow \mathbb{R}$ is called submodular if
\begin{align}
    f(S) + f(T) \geq f(S \cup T) + f(S \cap T),\ \forall S,T \subseteq [1:n].
\end{align}
A function $f: 2^{[1:n]} \rightarrow \mathbb{R}$ is called supermodular if $-f$ is submodular, and modular if it is both submodular and supermodular. 
\end{definition}
If $f(\phi) = 0,~f$ is modular if and only if $f(A)=\sum_{i\in A}f(\{i\})$, $\forall A\subseteq [1:n]$. For $S,T \subseteq [1:n]$, the conditional version of a submodular function $f$ is defined as $f(S|T) = f(S \cup T) - f(T)$ \cite{MadimanT10}.

\begin{definition}[Fractional Partition, Covering, Packing]
    \hspace{0.5em}
    \begin{enumerate}[leftmargin=*]
        \item Given a family $\mathcal{F}$ of subsets of $[1:n]$, a function $\gamma:~\mathcal{F}\rightarrow\mathbb{R}_{+}$ is called a fractional partition if, for all $i \in [1:n]$, $\sum\limits_{S \in \mathcal{F}: i \in S} \gamma(S) = 1$.
        \item Given a family $\mathcal{F}$ of subsets of $[1:n]$, a function $\alpha:~\mathcal{F}\rightarrow\mathbb{R}_{+}$ (resp. $\beta:~\mathcal{F}\rightarrow\mathbb{R}_{+}$) is called a fractional covering (resp. packing) if, for all $i \in [1:n]$, $\sum\limits_{S \in \mathcal{F}: i \in S} \alpha(S) \geq 1$ (resp. $\sum\limits_{S \in \mathcal{F}: i \in S} \beta(S) \leq 1$).        
    \end{enumerate}
\end{definition}

In this paper, we make use of the following inequalities for submodular functions, due to Madiman and Tetali~\cite{MadimanT10}.

\begin{theorem}[{\!\!\cite[Theorem~I]{MadimanT10}}]\label{thm:MT}
    Let $f: 2^{[1:n]} \rightarrow \mathbb{R}$ be any submodular function with $f(\phi) = 0$. Let $\gamma:\mathcal{F}\rightarrow \mathbb{R}_+$ be any fractional partition with respect to a family $\mathcal{F}$ of subsets of $[1:n]$. Then the following statements hold:
    \begin{enumerate}[leftmargin=*]
        \item {[{Strong form}]}
        \begin{align}\label{eq:strongMTIneq}
           \hspace{-0.5cm} \sum_{S \in \mathcal{F}} \gamma(S)f(S | S^{\text{c}} \setminus{> S})  \leq f([1:n]) \leq \sum\limits_{S \in \mathcal{F}} \gamma(S)f(S | < S).
        \end{align}

        \item {[{Weak form}]}
        \begin{align}\label{eq:MTIneq}
            \sum\limits_{S \in \mathcal{F}} \gamma(S)f(S | S^\emph{c}) \leq f([1:n]) \leq \sum\limits_{S \in \mathcal{F}} \gamma(S)f(S).
        \end{align}
    \end{enumerate}
    The fractional partition $\gamma$ in the lower and upper bounds can be replaced by fractional packing $\beta$ and fractional covering $\alpha$, respectively, if the submodular function $f$ is such that $f([1:j])$ is non-decreasing in $j$ for $j\in[1:n]$.
\end{theorem}
The upper bound in \eqref{eq:MTIneq}, with fractional covering $\alpha$ has been referred to as fractional subadditivity of submodular functions in the literature~\cite{Feige06}. There are two well-known special cases of the upper bound in \eqref{eq:MTIneq} for entropy:
\begin{enumerate}[leftmargin=*]
    \item Han's inequality~\cite{Han78,CoverT06} - $\mathcal{F} = \{\!\!\{ S \subseteq [1:n]: |S| = k\}\!\!\}, \gamma(S) = 1/\binom{n-1}{k-1}$.
    \item Shearer's Lemma~\cite{ChungGFS86,Radhakrishnan03} - For an arbitrary $\mathcal{F}$ and  $\alpha(S)=\frac{1}{k(\mathcal{F})},~\forall S \in \mathcal{F}$, where $k(\mathcal{F})$ denotes the maximum integer $k$ such that each $i\in[1:n]$ belongs to at least $k$ members of $\mathcal{F}$.
\end{enumerate}

\section{Generalized Information Inequalities via Submodularity}\label{sec:ineqsub}
In the following theorem, we present generalized information inequalities via submodularity, building on the information inequalities in Theorem~\ref{thm:MT}.
\begin{theorem}\label{SasonImprov}
    Let $\gamma$ be a fractional partition with respect to a family $\mathcal{F}$ of subsets of $[1:n]$. Let $f: 2^{[1:n]} \rightarrow \mathbb{R}$ be a submodular function with $f(\phi) = 0$, and $g: \mathbb{R} \rightarrow \mathbb{R}$ be a monotonically non-decreasing function.
    \begin{enumerate}[leftmargin=*]
        \item If $g$ is convex, then
            \begin{align}
                g\bigg( \frac{f([1:n])}{n} \bigg) \leq \sum\limits_{S \in \mathcal{F}} \frac{\gamma(S)|S|}{n} g\bigg(\frac{f(S | < S)}{|S|} \bigg). \label{eq: SasonIneq}
            \end{align}
        \item If $g$ is concave, then
            \begin{align}
                g\bigg( \frac{f([1:n])}{n} \bigg) \geq \sum\limits_{S \in \mathcal{F}} \frac{\gamma(S)|S|}{n} g\bigg(\frac{f(S | S^{\text{c}} \setminus{> S})}{|S|} \bigg). \label{eq: SLBIneq}
            \end{align}
    \end{enumerate}
\end{theorem}

\begin{remark}
    Both the assertions of Theorem~\ref{SasonImprov} continue to hold even if we replace $f$ by a supermodular (instead of submodular) and $g$ by a monotonically non-increasing function (instead of a non-decreasing function).
\end{remark}

\begin{remark}\label{fcremark}
    Assertions analogous to those of Theorem~\ref{SasonImprov} can also be obtained for the notions of fractional covering $\alpha$ and fractional packing $\beta$. Specifically, fractional covering and fractional packing may be used in place of a fractional partition in \eqref{eq: SasonIneq} and \eqref{eq: SLBIneq}, respectively, with appropriate normalization by their weights, provided that the submodular function $f$ is such that $f([1:j])$ is non-decreasing in $j$ for $j\in[1:n]$. These details are presented in
\if \extended 1%
    Appendix \ref{appendix:g}.
\fi
\if \extended 0%
    the extended version~\cite[Appendix~B]{JakharKCP26}.
\fi
\end{remark}

\begin{proof}[Proof Sketch of Theorem~\ref{SasonImprov}]
    Normalizing the submodular function values on both sides of the upper bound inequality in \eqref{eq:strongMTIneq} and observing that $\sum_{S \in \mathcal{F}} \frac{\gamma(S)|S|}{n} = 1$ since $\gamma$ is a fractional partition, part 1) follows by applying the convex function $g$ and invoking Jensen's inequality. Part 2) can be proved in a similar manner. A detailed proof is provided in
    \if \extended 1%
    Appendix~\ref{appendix:a}.
    \fi
    \if \extended 0%
    \cite[Appendix~A]{JakharKCP26}.
    \fi
\end{proof}

The bounds on $g\big( \frac{f([1:n])}{n} \big)$ in Theorem~\ref{SasonImprov} can be weakened by removing the conditioning term in \eqref{eq: SasonIneq} and adding more conditioning terms in \eqref{eq: SLBIneq}, respectively, keeping the assumptions on $f$ and $g$ intact. Formally, we get the following corollary.

\begin{corollary}\label{WeakSasonImprov}
    Let $\gamma$ be a fractional partition with respect to a family $\mathcal{F}$ of subsets of $[1:n]$. Let $f: 2^{[1:n]} \rightarrow \mathbb{R}$ be a submodular function with $f(\phi) = 0$, and $g: \mathbb{R} \rightarrow \mathbb{R}$ be a monotonically non-decreasing function.
    \begin{enumerate}[leftmargin=*]
        \item If $g$ is convex, then
            \begin{align}
                g\bigg( \frac{f([1:n])}{n} \bigg) \leq \sum\limits_{S \in \mathcal{F}} \frac{\gamma(S)|S|}{n} g\Big(\frac{f(S)}{|S|} \Big). \label{eq: WSasonIneq}
            \end{align}
            Moreover, if $g$ is strictly increasing, the equality in \eqref{eq: WSasonIneq} holds if and only if $g$ is linear in an interval containing $\big\{ \frac{f(S)}{|S|} : S \in \mathcal{F} \big\}$, and $f$ is modular.
        \item If $g$ is concave, then
            \begin{align}
                g\bigg( \frac{f([1:n])}{n} \bigg) \geq \sum\limits_{S \in \mathcal{F}} \frac{\gamma(S)|S|}{n} g\Big(\frac{f(S | S^{\text{c}})}{|S|} \Big). \label{eq: WSLBIneq}
            \end{align}
            Moreover, if $g$ is strictly increasing, the equality in \eqref{eq: WSLBIneq} holds if and only if $g$ is linear in an interval containing $\big\{ \frac{f(S | S^{\text{c}})}{|S|} : S \in \mathcal{F} \big\}$, and $f$ is modular. 
    \end{enumerate}
\end{corollary}

\begin{remark}\label{StrongHanRemark}
    When $\mathcal{F}=\{\!\!\{S\subseteq [1:n]: |S|=k\}\!\!\}$, $\gamma(S)=1/\binom{n-1}{k-1}$, the inequality $\eqref{eq: SasonIneq}$ yields a bound that is stronger than the corresponding result in \cite[Equation~22)]{Sason22}, while \eqref{eq: WSasonIneq} reduces exactly to it.
\end{remark}

\begin{remark}
    For $f(S)=H(X_S)$, $S\subseteq [1:n]$, and $g(x)=\mathsf{e}^{2x}$, the inequality \eqref{eq: WSasonIneq} recovers an entropy power-type inequality for joint distributions~\cite[Corollary~VI]{MadimanT10} while the stronger inequality \eqref{eq: SasonIneq} yields
    \begin{align}
        \mathsf{e}^{\frac{2H(X_{[1:n]})}{n}} \leq \sum\limits_{S \in \mathcal{F}} \frac{\gamma(S)|S|}{n} \mathsf{e}^{\frac{2H(X_S|X_{< S})}{|S|}},
    \end{align}
    a bound that is stronger than the former.
\end{remark}

\begin{remark}
         There is a duality between the upper and lower bounds in part 2) of Theorem~\ref{thm:MT}, through the corresponding inequality gaps, $\text{Gap}_{\text{L}}(f,\mathcal{F},\gamma) = f([1:n]) - \sum_{S \in \mathcal{F}} \gamma(S)f(S | S^\text{c})$ and $\text{Gap}_{\text{U}}(f,\mathcal{F},\gamma) = \sum\limits_{S \in \mathcal{F}} \gamma(S)f(S)-f([1:n])$~\cite[Theorem~IV]{MadimanT10}, in particular, 
         \begin{align}
             \frac{\text{Gap}_{\text{L}}(f,\mathcal{F},\gamma)}{w(\gamma)}=\frac{\text{Gap}_{\text{U}}(f,\bar{\mathcal{F}},\bar{\gamma})}{w(\bar{\gamma})},
         \end{align}
         where $\bar{\mathcal{F}}=\{\!\!\{S^\text{c}: S\in\mathcal{F}\}\!\!\}$, $w(\gamma)=\sum_{S\in\mathcal{F}}\gamma(S)$, and $\bar{\gamma}(S^\text{c})=\frac{\gamma(S)}{w(\gamma)-1}$. In contrast, Theorem~\ref{SasonImprov} (and even Corollary~\ref{WeakSasonImprov}) does not, in general, admit such a duality due to the presence of the potentially nonlinear function $g$.
\end{remark}

\begin{proof}[Proof Sketch of Corollary~\ref{WeakSasonImprov}]
    The inequalities \eqref{eq: WSasonIneq} and \eqref{eq: WSLBIneq} follow from the fact that conditioning reduces the value of a submodular function, together with the monotonicity of $g$. The equality conditions follow from the equality conditions for Jensen's inequality and those for the weak form of the Madiman and Tetali's inequalities~\cite[Theorem~4]{JakharKCP25}. 
    A detailed proof is provided in
    \if \extended 1%
    Appendix~\ref{appendix:b}.
    \fi
    \if \extended 0%
    \cite[Appendix~C]{JakharKCP26}.
    \fi
\end{proof}
We remark that the results in this section can be particularly useful in scenarios where the quantities of interest and the observable data are convex transformations of normalized values of a submodular function, and where recovering the underlying submodular values is either infeasible or computationally inefficient.

\section{Two Combinatorial Problems}
In this section, we consider two combinatorial problems: a projection inequality of the Loomis-Whitney-type~\cite{LoomisH49,Radhakrishnan03}, and an extremal graph theory problem which is motivated by the formulation introduced in~\cite[Section~5]{Sason22}.
\subsection{A Strong Loomis-Whitney-Type Inequality}

A well-known projection inequality, the Loomis-Whitney inequality~\cite{LoomisH49}, relates the cardinality of a set in $\mathbb{R}^3$ to the cardinalities of its two dimensional coordinate projections. Specifically, for a set of $n$ points in $\mathbb{R}^3$, let $n_{xy}$, $n_{yz}$, and $n_{zx}$ denote the number of distinct points in its projections onto the $xy$-, $yz$-, and $zx$-planes, respectively. Then $n^2\leq n_{xy}n_{yz}n_{zx}$. An elegant entropy-based proof of this inequality was given by Radhakrishnan~\cite{Radhakrishnan03} using Han's inequality, a special case of \eqref{eq:MTIneq}. Here, we use the strong form of the corresponding inequality in \eqref{eq:strongMTIneq} to derive a Loomis-Whitney-type inequality of similar flavour. 

\begin{theorem}\label{StrongLW}
   Let $S\subset \mathbb{R}^3$ be a set of $n$ distinct points. Let $n_{xy}$ and $n_{zx}$ denote the number of distinct points in the projections
   of $S$ onto the $xy$- and $zx$-planes, respectively. Define 
   \begin{align}
       n^*_{yz|x}=\max_{x\in\mathbb{R}}\left|\{(y,z)\in\mathbb{R}^2:(x,y,z)\in S\}\right|,
   \end{align}
   that is, the maximum number of distinct projection points onto the $yz$-plane among all points of $S$ having the same $x$-coordinate. Then
   \begin{align}\label{eq:StrongLW}
       n^2\leq n_{xy}n^*_{yz|x}n_{zx}.
   \end{align}
\end{theorem}

\begin{remark}
    The inequality in \eqref{eq:StrongLW} strengthens the bound $n_{xy}n_{yz}n_{zx}$ in Loomis-Whitney inequality~\cite{LoomisH49,Radhakrishnan03} by replacing the cardinality of the global projection onto the $yz$-plane with the maximum cardinality of a slice at a fixed $x$-coordinate, noting that $n^*_{yz|x}\leq n_{yz}$. This strengthening relies on additional slice-level structural information beyond coordinate projections, inherited from the strong form of the underlying entropy inequality.
\end{remark}
A detailed proof of Theorem~\ref{StrongLW} is provided in 
 \if \extended 1%
    Appendix~\ref{appendix:c}.
    \fi
    \if \extended 0%
    \cite[Appendix~D]{JakharKCP26}.
    \fi
\begin{example}
    Consider the set of points in $\mathbb{R}^3$, 
    \begin{align*}
        S = \{(1,1,2),(1,2,2),(2,1,2),(2,3,2),(2,2,1),(2,2,3)\}.
    \end{align*}
    The number of distinct projection points onto the coordinate planes are  $n_{xy} = 5$, $n_{yz} = 5$, and $n_{zx} = 4$. Consequently, the Loomis-Whitney inequality yields the bound 
    \begin{align}
    n^2\leq 5\times5\times4 = 100.
    \end{align}
    On the other hand, the maximum of distinct projection points onto the $yz$-plane among all points of $S$ having the same $x$-coordinate is $n^*_{yz|x}=4$, attained by the slice corresponding to $x=2$. Therefore, the bound in \eqref{eq:StrongLW} becomes 
    \begin{align}
        n^2\leq 5\times4\times4=80,
    \end{align}
    which is strictly smaller than the Loomis-Whitney bound. The comparison is illustrated in the following figure.
\end{example}

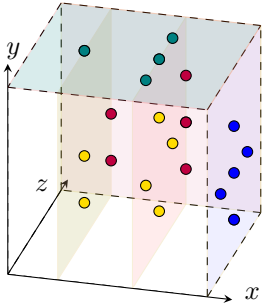
\begin{figure}[H]
\centering
    \begin{tikzpicture}
    \begin{axis}[
        view={15}{25},
        axis equal image,
        axis lines=center,
        xmin=0, xmax=4.5,
        ymin=0, ymax=4.5,
        zmin=0, zmax=4.5,
        grid=both,
        grid style={draw=gray!30},
        set layers,
        axis on top=false,
        axis background/.style={fill=none},
        x axis plane/.style={fill=none},
        y axis plane/.style={fill=none},
        z axis plane/.style={fill=none},
        xlabel={$x$},
        ylabel={$y$},
        zlabel={$z$},
        xtick=\empty,
        ytick=\empty,
        ztick=\empty,
        xlabel style={at={(axis description cs:0.8,0.02)}, anchor=west},
        ylabel style={at={(axis description cs:0.02,0.7)}, anchor=south},
        zlabel style={at={(axis description cs:0.12,0.3)}, anchor=south}
    ]
    
    \addplot3[black, dashed] coordinates {(0,0,0)(4,0,0)};
    \addplot3[black, dashed] coordinates {(4,0,0)(4,4,0)};
    \addplot3[black, dashed] coordinates {(4,4,0)(0,4,0)};
    \addplot3[black, dashed] coordinates {(0,4,0)(0,0,0)};
    
    \addplot3[black, dashed] coordinates {(0,0,4)(4,0,4)};
    \addplot3[black, dashed] coordinates {(4,0,4)(4,4,4)};
    \addplot3[black, dashed] coordinates {(4,4,4)(0,4,4)};
    \addplot3[black, dashed] coordinates {(0,4,4)(0,0,4)};
    
    \addplot3[black, dashed] coordinates {(0,0,0)(0,0,4)};
    \addplot3[black, dashed] coordinates {(4,0,0)(4,0,4)};
    \addplot3[black, dashed] coordinates {(4,4,0)(4,4,4)};
    \addplot3[black, dashed] coordinates {(0,4,0)(0,4,4)};

    \addplot3[
        patch,patch type=rectangle,
        fill=olive!85,
        fill opacity=0.15,
        opacity=0.15,
        draw=none
    ]
    coordinates {(1,0,0)(1,4,0)(1,4,4)(1,0,4)};
    
    \addplot3[
        patch,patch type=rectangle,
        fill=red!50,
        fill opacity=0.15,
        opacity=0.15,
        draw=none
    ]
    coordinates {(2.5,0,0)(2.5,4,0)(2.5,4,4)(2.5,0,4)};

    \addplot3[
        patch,patch type=rectangle,
        fill=blue!40,
        fill opacity=0.15,
        opacity=0.15,
        draw=none
    ]
    coordinates {(4,0,0)(4,4,0)(4,4,4)(4,0,4)};
    
    \addplot3[
        patch,patch type=rectangle,
        fill=purple!40,
        fill opacity=0.15,
        opacity=0.15,
        draw=none
    ]
    coordinates {(0,4,0)(4,4,0)(4,4,4)(0,4,4)};
    
    \addplot3[
        patch,patch type=rectangle,
        fill=teal,
        fill opacity=0.15,
        opacity=0.15,
        draw=none
    ]
    coordinates {(0,0,4)(4,0,4)(4,4,4)(0,4,4)};

    \addplot3[
        only marks,
        mark=*,
        mark size=2,
        draw=black,
        fill=yellow!80!orange
    ]
    coordinates {
        (1,2,0.75)
        (1,2,1.75)
        (2.5,2,2.75)
        (2.5,2,0.75)
        (2.5,1,1.75)
        (2.5,3,1.75)
    };

    \addplot3[only marks,mark=*,mark size=2,draw=black,fill=blue]
    coordinates {
        (4,2,0.75)
        (4,2,1.75)
        (4,2,2.75)
        (4,2,0.75)
        (4,1,1.75)
        (4,3,1.75)
    };
    
    \addplot3[only marks,mark=*,mark size=2,draw=black,fill=purple]
    coordinates {
        (1,4,0.75)
        (1,4,1.75)
        (2.5,4,2.75)
        (2.5,4,0.75)
        (2.5,4,1.75)
        (2.5,4,1.75)
    };
    
    \addplot3[only marks,mark=*,mark size=2,draw=black,fill=teal]
    coordinates {
        (1,2,4)
        (1,2,4)
        (2.5,2,4)
        (2.5,2,4)
        (2.5,1,4)
        (2.5,3,4)
    };
    \end{axis}
    \end{tikzpicture}
\caption{Illustration of the comparison between the classical Loomis–Whitney bound and the strong Loomis–Whitney–type bound in $\mathbb{R}^3$. The yellow points represent the set $S$. The red, blue, and green points denote the projections of $S$ onto the $xy$-, $yz$-, and $zx$-planes, respectively. The slice at $x=2$ has the largest number of distinct $yz$-projections, leading to a tighter bound than the Loomis-Whitney bound.}
\end{figure}

Using the same approach, we obtain a $d$-dimensional analogue of the above strong Loomis-Whitney-type inequality.

\begin{theorem}\label{StrongLW-ddim}

 Let $S \subset \mathbb{R}^d$ be a set of n distinct points. For $i\in[2:d]$, let $n_{-i}$ denote the number of distinct points in the projection of $S$ on to the $(d-1)$-dimensional hyperplane obtained by deleting the $i^{\text{th}}$ coordinate. Define 
 \begin{align}
        n^*_{-1|x_1} = \max_{x_1\in\mathbb{R}}\big|\{(x_2, \ldots, x_d):(x_1,x_2, \ldots, x_d)\in S\}\big|,
    \end{align}
    that is, the maximum number of distinct projection points onto the $(d-1)$-dimensional coordinate hyperplane obtained by deleting the $x_1$-coordinate among all points of $S$ having the same $x_1$-coordinate. Then
    \begin{align}\label{eq:StrongLW-ddim}
        n^{d-1}\leq \left(\prod_{i=2}^{d}n_{-i}\right)n^*_{-1|x_1}.
    \end{align}
    \end{theorem}

\begin{proof}[Proof Sketch of Theorem~\ref{StrongLW-ddim}]
    We choose a point $P=(X_1, \ldots, X_d)$ uniformly at random from $S$. We apply the upper bound in inequality \eqref{eq:strongMTIneq} to $H(X_1, \ldots, X_d)$ by taking $f(S) = H(X_S), S \subseteq [1:n]$, $\mathcal{F} = \{\!\!\{[1:d]\setminus{\{i\}}: i \in [1:d]\}\!\!\}$ and $\gamma(S) = 1/(d-1), ~\forall S \in \mathcal{F}$. Since $P$ is uniformly distributed over $S$, we have $H(X_1, \ldots, X_d)=\log{n}$. Using the structural information in the theorem statement to upper bound each entropy term in the upper bound in \eqref{eq:strongMTIneq}, we get \eqref{eq:StrongLW-ddim}. A detailed proof is provided in
    \if \extended 1%
    Appendix~\ref{appendix:d}.
    \fi
    \if \extended 0%
    \cite[Appendix~E]{JakharKCP26}.
    \fi
\end{proof}
The inequality in \eqref{eq:StrongLW-ddim} strengthens the bound in Loomis-Whitney bound in $d$ dimensions~\cite{LoomisH49},\cite[Theorem~3.4]{Galvin14}. Notice that we can obtain different inequalities corresponding to \eqref{eq:StrongLW-ddim} depending on the coordinate for which we have the additional slice-level structural information, using a different ordering of the coordinates. This is a consequence of the fact that the strong form of Han's inequality holds for any ordering of the ground set.

\subsection{An Extremal Graph Theory Problem}
We now present an application of Shearer's lemma to an extremal graph theory problem. Specifically, we are interested in estimating the number of edges in an undirected graph $G$ with vertex set $V(G)\subseteq\{-1,1\}^n$ and edge set $E(G)$ defined as follows. For distinct vertices $x^n,y^n$, the edge $\{x^n,y^n\}$ belongs to $E(G)$ if and only if the set of coordinates in which $x^n$ and $y^n$ differ belongs to a prescribed family $\mathcal{F}$ of subsets of $[1:n]$, i,e.,
\begin{align}\label{edgedefn}
    \{x^n,y^n\}\in E(G) \Leftrightarrow \Delta(x^n,y^n)\in \mathcal{F},
\end{align}
where $\Delta(x^n,y^n)=\{i\in[1:n]: x_i\neq y_i\}$.

Special cases of this setup have been studied previously \cite{BoucheronLM13,Sason22}. In particular, choosing $\mathcal{F}=\{\!\!\{ S\subseteq [1:n]: |S|\leq \tau\}\!\!\}$ recovers the problem studied by Sason~\cite{Sason22}, while the choice $\mathcal{F}=\{\!\!\{S\subseteq [1:n]: |S|=1\}\!\!\}$ corresponds to the problem considered by Boucheron \emph{et al.}~\cite{BoucheronLM13}. 

Building on the confusability graph viewpoint for the model in \cite{Sason22}, our setup admits a similar but a more general interpretation. Consider a noisy channel with binary $n$-length sequences as inputs. Two sequences $x^n,y^n\in V(G)$ may produce a common channel output with positive probability if and only if $\Delta(x^n,y^n)\in\mathcal{F}$. Note that in the models considered in \cite{BoucheronLM13} and \cite{Sason22}, confusability depends only on the Hamming distances due to the specific choices of $\mathcal{F}$, whereas in our setting it potentially depends on the full pattern of coordinates in which they differ. This distinction is reflected in our analysis through the use of Shearer’s lemma, rather than Han’s inequality, which suffices set families consisting of all sets of equal sizes as considered in \cite{BoucheronLM13} and \cite{Sason22}.

We derive an upper bound on the number of edges, $|E(G)|$, in terms of the vertex set $V(G)$. To this end, we need the following definitions. For any $x^n = (x_1,\ldots,x_n) \in \{-1,1\}^n, d \in [1:n]$ and $S = \{i_1, \ldots, i_d\} \subseteq [1:n]$ such that $i_1 < \ldots < i_d$, let $\Tilde{x}^{(S)}$ denote an $n-d$ length vector obtained by dropping the bits corresponding to positions $i_1,\ldots,i_d$ in $x^n$, i.e.,
\begin{align}
    \Tilde{x}^{(S)} := (x_1,\ldots,x_{i_1-1},x_{i_1+1},\ldots,x_{i_d-1},x_{i_d+1},\ldots,x_n),
\end{align}
and let $\Bar{x}^{(S)}$ denote an $n$ length vector obtained by flipping the bits corresponding to positions $i_1,\ldots,i_d$ in $x^n$, i.e., $\Bar{x}^{(S)}_j = -x_j$, for $j \in S$ and $\Bar{x}^{(S)}_j = x_j$ otherwise. Further, for $d \in [1:n]$, we define two integers $m_d$ and $\ell_d$ as follows.
\begin{align}\label{eq: m_d}
    m_d := \hspace{-0.2em}\min\limits_{\substack{x^n \in V(G),\\S \in \mathcal{F}:\\ |S| = d}} \left|\left\{ y^n \in V(G) : \Tilde{y}^{(S)} = \Tilde{x}^{(S)}, \Bar{x}^{(S)} \in V(G) \right\}\right|,    
\end{align}
\begin{align}\label{eq: l_d}
    \ell_d := \hspace{-0.2em} \min\limits_{\substack{x^n \in V(G),\\S \in \mathcal{F}:\\ |S| = d}} \left|\left\{ y^n \in V(G) : \Tilde{y}^{(S)} = \Tilde{x}^{(S)}, \Bar{x}^{(S)} \notin V(G) \right\}\right|. 
\end{align}
By definitions, $m_d$ and $\ell_d$ satisfy $m_d \in [2: \min{\{2^{d},|V(G)|\}}]$ and $\ell_d \in [1: \min{\{2^{d}-1,|V(G)|-1\}}]$ respectively.
    
\begin{theorem}\label{ExtremalGraph}
    Let $G$ be a simple undirected graph with set of vertices $V(G) \subseteq \{-1,1\}^n, n \in \mathbb{N}$ and set of edges $E(G)$ connecting pairs of vertices if and only if the set of coordinates in which the vertices differ is an element in the family $\mathcal{F}$ of subsets of $[1:n]$. Let $r = \max_{S \in \mathcal{F}} |S|$. For $d \in [1:r]$, let $\mathcal{F}_d$ denote the collection of sets $S \in \mathcal{F}$ such that $|S| = d$, and let integers $m_d$ and $\ell_d$ be as defined in \eqref{eq: m_d} and \eqref{eq: l_d} respectively, then the following holds:
    \begin{align}
        |E(G)| \leq \sum\limits_{d = 1}^r\frac{|V(G)|\Big( (|\mathcal{F}_d| - t_d) \log{|V(G)|} - |\mathcal{F}_d| \log{\ell_d}\Big)}{2\log{\frac{m_d}{\ell_d}}}, \label{eq: EGBound}
    \end{align}
    where $t_d$ denotes the maximum integer such that each $i \in [1:n]$ appears in at least $t_d$ members of $\Bar{\mathcal{F}_d} = \{\!\!\{S^\emph{c}: S\in\mathcal{F}_d\}\!\!\}$.
\end{theorem}

\begin{remark}\label{BLMSasonRemark}
    Theorem~\ref{ExtremalGraph} recovers the bounds obtained in \cite[Theorem~4.2]{BoucheronLM13} and \cite[Theorem~2]{Sason22} for the choices $\mathcal{F} = \{\!\!\{ \{i\}: i \in [1:n] \}\!\!\}$ and $\mathcal{F}=\{\!\!\{ S \subseteq [1:n]: |S| \leq \tau \}\!\!\}$, respectively. See
     \if \extended 1%
    Appendix~\ref{appendix:f}
    \fi
    \if \extended 0%
    \cite[Appendix~G]{JakharKCP26}
    \fi
    for details.
\end{remark}

\begin{proof}[Proof Sketch of Theorem~\ref{ExtremalGraph}]
    We choose a vertex $X^n = (X_1, \ldots, X_n)$ uniformly at random from $V(G)$. Noting that for any $x^n \in V(G)$, either $\Bar{x}^{(S)} \in V(G)$ or $\Bar{x}^{(S)} \notin V(G)$, we obtain upper bounds on $P_{X_S|\Tilde{X}^{(S)}}(x_S|\Tilde{x}^{(S)})$ under both the cases as $\frac{1}{m_d}$ and $\frac{1}{\ell_d}$, respectively. Using these bounds, we obtain a lower bound on $H(X_S|\Tilde{X}^{(S)}) = H(X^n)-H(\Tilde{X}^{(S)})$ and consequently on $\sum_{S \in \mathcal{F}_d} \left(H(X^n)-H(\Tilde{X}^{(S)})\right)$ as
    \begin{align}
        &\sum\limits_{S \in \mathcal{F}_d} \left(H(X^n) - H(\Tilde{X}^{(S)})\right) \nonumber\\
        & \geq \frac{\log{m_d}}{|V(G)|}2|E_{\mathcal{F}_d}(G)| + \frac{\log{\ell_d}}{|V(G)|} \big(|\mathcal{F}_d||V(G)| - 2|E_{\mathcal{F}_d}(G)|\big), \label{eq: EGPS_2}
    \end{align}
    where $E_{\mathcal{F}_d}(G)$ denotes the set of edges connecting pairs of vertices $x^n,y^n\in V(G)$ such that $\Delta(x^n,y^n)\in \mathcal{F}_d$. Finally, we obtain an upper bound as
    \begin{align}
       \! \!\!\sum_{S \in \mathcal{F}_d} \!\left(H(X^n) - H(\Tilde{X}^{(S)})\right) &=|\mathcal{F}_d|H(X^n)-\!\sum_{S^\text{c} \in \Bar{\mathcal{F}}_d}\!H(\Tilde{X}^{(S)})\\
        &\leq (|\mathcal{F}_d| - t_d)H(X^n)\label{eqn:Shearerapl}\\
        & = (|\mathcal{F}_d| - t_d) \log{|V(G)|}, \label{eq: EGPS_3}
    \end{align}
    where \eqref{eqn:Shearerapl} follows by an application of Shearer's lemma for the family $\Bar{\mathcal{F}_d} = \{\!\!\{S^\emph{c}: S\in\mathcal{F}_d\}\!\!\}$, and \eqref{eq: EGPS_3} follows from the fact that $X^n$ is uniformly distributed over $V(G)$. Note that although $\Bar{\mathcal{F}_d}$ consists of sets all having the same cardinality $n-d$, it does not contain all possible subsets of that size due to the arbitrary structure of $\mathcal{F}$. So, Han's inequality for families consisting of all sets of fixed size is not applicable here.
     Finally, we put together \eqref{eq: EGPS_2} and \eqref{eq: EGPS_3} to obtain
    \begin{align}
        |E_{\mathcal{F}_d}(G)| \leq \frac{|V(G)|\Big( (|\mathcal{F}_d| - t_d) \log{|V(G)|} - |\mathcal{F}_d| \log{\ell_d}\Big)}{2\log{\frac{m_d}{\ell_d}}}. \label{eq: EGPS_4}
    \end{align}
    Consequently, summing \eqref{eq: EGPS_4} over $d \in [1:r]$ on both the sides of \eqref{eq: EGPS_4} gives an upper bound on $|E(G)|$ given in \eqref{eq: EGBound}, since $|E(G)| = \sum_{d = 1}^r |E_{\mathcal{F}_d}(G)|$ by the definitions of $E(G)$ and $E_{\mathcal{F}_d}(G)$. A detailed proof is provided in
    \if \extended 1%
    Appendix~\ref{appendix:e}.
    \fi
    \if \extended 0%
    \cite[Appendix~F]{JakharKCP26}.
    \fi
\end{proof}

\begin{example}
     Let $\mathcal{F} = \{\!\!\{ \{1\}, \{2\}, \{3\}, \{4\}, \{5\}, \{1,3\}, \{1,4\},\\ \{3,4\}, \{2,5\}, \{1,3,4\} \}\!\!\}$. Consider a graph $G$ with vertex set $V(G) \subseteq \{-1,1\}^5$ and edge set $E(G)$ defined according to \eqref{edgedefn} with this choice of $\mathcal{F}$. We have
    \begin{align}
    |\mathcal{F}_1| = 5, t_1 = 4;|\mathcal{F}_2| = 4, t_2 = 2;|\mathcal{F}_3| = 1, t_3 = 0.
    \end{align}
    We use the minimum admissible values of $m_d = 2$ and $\ell_d = 1$ to evaluate the upper bound in \eqref{eq: EGBound},
    \begin{align}
       |E(G)| & = \sum\limits_{d = 1}^3\frac{|V(G)|\Big( (|\mathcal{F}_d| - t_d) \log{|V(G)|}\Big)}{2\log{2}} \\
        & = \frac{1}{2}|V(G)|\log{|V(G)|}(5 - 4) \nonumber \\
        & \quad+ \frac{1}{2}|V(G)|\log{|V(G)|}(4 - 2) \nonumber \\
        & \quad + \frac{1}{2}|V(G)|\log{|V(G)|}(1 - 0) \\
        & = 2|V(G)|\log{|V(G)|}.
    \end{align}
\end{example}

\if \extended 1
\IEEEtriggeratref{14}
\fi
\if \extended 0
\IEEEtriggeratref{15}
\fi

\bibliographystyle{IEEEtran}
\bibliography{bibliofile}

\clearpage

\if \extended 1

\appendices
\section{Proof of Theorem~\ref{SasonImprov}}\label{appendix:a}
\emph{Proof of Part 1).} Using the upper bound on $f([1:n])$ from the strong form in Theorem~\ref{thm:MT}, we have that
\begin{align}
    f([1:n]) \leq \sum\limits_{S \in \mathcal{F}} \gamma(S) f(S | < S). \label{eq: si_1.1}
\end{align}
Dividing both the sides by $n$, we get
\begin{align}
    \frac{f([1:n])}{n} & \leq \frac{1}{n}\sum\limits_{S \in \mathcal{F}} \gamma(S) f(S | < S) \\
    & = \sum\limits_{S \in \mathcal{F}} \gamma(S) \frac{|S|}{n} \frac{f(S | < S)}{|S|},
\end{align}
and, since $g$ is monotonically non-decreasing,
\begin{align}
    g\bigg( \frac{f([1:n])}{n} \bigg) & \leq g\bigg( \sum\limits_{S \in \mathcal{F}} \gamma(S) \frac{|S|}{n} \frac{f(S | < S)}{|S|} \bigg) \label{eq: si_1.2}.
\end{align}
Notice that 
\begin{align}
    \sum\limits_{S \in \mathcal{F}} \frac{\gamma(S)|S|}{n} & = \sum\limits_{S \in \mathcal{F}} \sum\limits_{i \in S} \frac{\gamma(S)}{n} \\
    & = \sum\limits_{i = 1}^n \sum\limits_{\substack{S \in \mathcal{F}: \\ i \in S}} \frac{\gamma(S)}{n} \label{eq: si_1.4} \\
    & = \sum\limits_{i = 1}^n \frac{1}{n} = 1, \label{eq: si_1.5}
\end{align}
where \eqref{eq: si_1.4} follows by interchanging the summations and \eqref{eq: si_1.5} holds because $\gamma$ is a fractional partition, i.e., $\sum\limits_{\substack{S \in \mathcal{F}: \\ i \in S}} \gamma(S) = 1$. Now, since $g$ is convex, Jensen's inequality applied to \eqref{eq: si_1.2} implies that
\begin{align}
    g\bigg( \frac{f([1:n])}{n} \bigg) \leq \sum\limits_{S \in \mathcal{F}} \frac{\gamma(S)|S|}{n} g\bigg( \frac{f(S | < S)}{|S|} \bigg). \label{eq: si_1.3}
\end{align}

\emph{Proof of Part 2).} The proof for part 2 closely resembles that of part 1; however, we will provide the details for completeness. We start with the lower bound on $f([1:n])$ from the strong form in Theorem~\ref{thm:MT}, in place of the upper bound in \eqref{eq: si_1.1}, i.e.,
\begin{align}
    f([1:n]) \geq \sum\limits_{S \in \mathcal{F}} \gamma(S)f(S | S^{\text{c}} \setminus{> S}) \label{eq: si_2.1}
\end{align}
Dividing both the sides by $n$, we get,
\begin{align}
    \frac{f([1:n])}{n} & \geq \frac{1}{n}\sum\limits_{S \in \mathcal{F}} \gamma(S)f(S | S^{\text{c}} \setminus{> S}) \\
    & = \sum\limits_{S \in \mathcal{F}} \gamma(S) \frac{|S|}{n} \frac{f(S | S^{\text{c}} \setminus{> S})}{|S|},
\end{align}
and, since $g$ is monotonically non-decreasing,
\begin{align}
    g\bigg( \frac{f([1:n])}{n} \bigg) & \geq g\bigg( \sum\limits_{S \in \mathcal{F}} \gamma(S) \frac{|S|}{n} \frac{f(S | S^{\text{c}} \setminus{> S})}{|S|} \bigg) \label{eq: si_2.2}.
\end{align}

Since $g$ is concave and $\sum\limits_{S \in \mathcal{F}} \frac{\gamma(S)|S|}{n} = 1$, Jensen's inequality applied to \eqref{eq: si_2.2} implies that
\begin{align}
    g\bigg( \frac{f([1:n])}{n} \bigg)  & \geq \sum\limits_{S \in \mathcal{F}} \frac{\gamma(S)|S|}{n} g\bigg( \frac{f(S | S^{\text{c}} \setminus{> S})}{|S|} \bigg). \label{eq: si_2.3}
\end{align}

\section{Fractional covering/packing version for Theorem~\ref{SasonImprov}}\label{appendix:g}
\begin{theorem}\label{SasonCovering}
    Let $\alpha$ and $\beta$ be any fractional covering, and fractional packing, respectively, with respect to a family $\mathcal{F}$ of subsets of $[1:n]$. Let $f: 2^{[1:n]} \rightarrow \mathbb{R}$ be a submodular function with $f(\phi) = 0$, such that $f([1:j])$ is non-decreasing in $j$ for $j\in[1:n]$, and $g: \mathbb{R} \rightarrow \mathbb{R}$ be a monotonically non-decreasing function.
    \begin{enumerate}[leftmargin=*]
        \item If $g$ is convex, then
            \begin{align}
                g\bigg( \frac{f([1:n])}{n\cdot v(\alpha)} \bigg) \leq \sum\limits_{S \in \mathcal{F}} \frac{\alpha(S)|S|}{n\cdot v(\alpha)} g\bigg( \frac{f(S | < S)}{|S|} \bigg),
            \end{align}
            where $v(\alpha) = \sum_{S \in \mathcal{F}} \frac{\alpha(S)|S|}{n}$.
        \item If $g$ is concave, then
            \begin{align}
                g\bigg( \frac{f([1:n])}{n\cdot v'(\beta)} \bigg) \geq \sum\limits_{S \in \mathcal{F}} \frac{\beta(S)|S|}{n\cdot v'(\beta)} g\bigg( \frac{f(S | S^{\text{c}} \setminus{> S})}{|S|} \bigg), \label{eq: fcrm3}
            \end{align}
            where $v'(\beta) = \sum_{S \in \mathcal{F}} \frac{\beta(S)|S|}{n}$.
    \end{enumerate}
\end{theorem}
\begin{proof}
    Since, the submodular function $f$ is such that $f([1:j])$ is non-decreasing in $j$ for $j\in[1:n]$, it follows from Theorem~\ref{thm:MT} that
    \begin{align}
        f([1:n]) \leq \sum\limits_{S \in \mathcal{F}} \alpha(S) f(S | < S).
    \end{align}
    Dividing both the sides by $n\cdot v(\alpha)$, we get
    \begin{align}
        \frac{f([1:n])}{n\cdot v(\alpha)} & \leq \frac{1}{n\cdot v(\alpha)}\sum\limits_{S \in \mathcal{F}} \alpha(S) f(S | < S) \\
        & = \sum\limits_{S \in \mathcal{F}} \frac{\alpha(S)|S|}{n\cdot v(\alpha)} \frac{f(S | < S)}{|S|},
    \end{align}
    and, since $g$ is monotonically non-decreasing,
    \begin{align}
        g\bigg( \frac{f([1:n])}{n\cdot v(\alpha)} \bigg) & \leq g\bigg( \sum\limits_{S \in \mathcal{F}} \frac{\alpha(S)|S|}{n\cdot v(\alpha)} \frac{f(S | < S)}{|S|} \bigg) . \label{eq: fcrm1}
    \end{align}
    Notice that $\sum\limits_{S \in \mathcal{F}} \frac{\alpha(S)|S|}{n\cdot v(\alpha)} = 1$. Now, since $g$ is convex, Jensen's inequality applied to \eqref{eq: fcrm1} implies that
    \begin{align}
        g\bigg( \frac{f([1:n])}{n\cdot v(\alpha)} \bigg) \leq \sum\limits_{S \in \mathcal{F}} \frac{\alpha(S)|S|}{n\cdot v(\alpha)} g\bigg( \frac{f(S | < S)}{|S|} \bigg).
    \end{align}
    
    We can also show that \eqref{eq: fcrm3} holds when $g$ is monotonically non-decreasing and concave, by making use of the lower bound from part 1) of Theorem~\ref{thm:MT} with a fractional packing and proceeding in a manner analogous to the above argument.
\end{proof}

\section{Proof of Corollary~\ref{WeakSasonImprov}}\label{appendix:b}
\emph{Proof of Part 1).} The assertion in part 1) follows from part 1) of Theorem~\ref{SasonImprov} by noting that $f(S|<S) \leq f(S)$ as a consequence of the submodularity of $f$, together with the fact that $g$ is monotonically non-decreasing.

To analyze the necessary and sufficient conditions for equality to hold in \eqref{eq: WSasonIneq} when $g$ is strictly increasing and convex, we refer to the proof of part 1) of Theorem~\ref{SasonImprov}, with $f(S|<S)$ replaced by $f(S)$. Clearly, equality holds if and only if equality is attained both in Jensen's inequality and in the weak form of Madiman and Tetali's inequalities, i.e., in \eqref{eq: si_1.1} with  $f(S|<S)$ replaced by $f(S)$ hold. This occurs if and only if 
$g$ is linear in an interval containing $\big\{ \frac{f(S)}{|S|} : S \in \mathcal{F} \big\}$, and $f$ is modular, the latter follows from\cite[Theorem~4]{JakharKCP25}.

\emph{Proof of Part 2).} The proof of part 2) follows from part 2) of Theorem~\ref{SasonImprov}, since $f(S|S^{\text{c}}\setminus{>S}) \geq f(S|S^{\text{c}})$ as a consequence of the submodularity of $f$. The equality conditions can be characterized in a manner analogous to that of part~1).

\section{Proof of Theorem~\ref{StrongLW}}\label{appendix:c}
We choose $(X_1, X_2, X_3)$ uniformly at random from $S$. We then apply \eqref{eq:strongMTIneq} with $f(S) = H(X_S), S \subseteq [1:3]$, $\mathcal{F} = \{\!\!\{\{1,2\}, \{2,3\}, \{3,1\}\}\!\!\}$, and $\gamma(S) = 1/2, ~\forall S \in \mathcal{F}$, to obtain an upper bound on $H(X_1,X_2,X_3)$ as follows.
\begin{align}
    & H(X_1,X_2,X_3) \nonumber \\
    & \leq \sum\limits_{S \in \mathcal{F}} \gamma(S)H(X_S|X_{< S}) \\
    & = \frac{1}{2}(H(X_1, X_2) + H(X_2, X_3|X_1) + H(X_3, X_1)) \\
    & \leq \frac{1}{2} (\log{n_{xy}} + H(X_2, X_3|X_1) + \log{n_{zx}}) \label{eq: SLW_1}\\
    & \leq \frac{1}{2} (\log{n_{xy}} + \log{n^*_{yz|x}} + \log{n_{zx}}), \label{eq: SLW_2}
\end{align}
where \eqref{eq: SLW_1} follows from the definitions of $n_{xy}$ and $n_{zx}$, together with the fact that $H(W) \leq \log{|\mathcal{W}|}$, where $\mathcal{W}$ denotes the alphabet of $W$. \eqref{eq: SLW_2} follows because
\begin{align}
    &H(X_2,X_3|X_1) \nonumber \\
    & = \sum\limits_{x_1} P_{X_1}(x)H(X_2,X_3|X_1 = x_1) \\
    & \leq \sum\limits_{x_1} P_{X_1}(x_1) \log{(|\{(y,z)\in\mathbb{R}^2:(x_1,y,z)\in S\}|)}\label{eq:SLW3prime} \\
    & \leq \sum\limits_{x_1} P_{X_1}(x_1) \log{n^*_{yz|x}} \label{eq: SLW_3}\\
    & = \log{n^*_{yz|x}},
\end{align}
where \eqref{eq:SLW3prime} holds since, conditioned on $X_1=x_1$, the alphabet of $(X_2,X_3)$ is given by $\{(y,z)\in\mathbb{R}^2:(x_1,y,z)\in S\}$ and  \eqref{eq: SLW_3} follows from the definition of $n^*_{yz|x}$. Combining \eqref{eq: SLW_2} with the fact that $H(X_1,X_2,X_3) = \log{n}$, we get
\begin{align}
    n^2\leq n_{xy}n^*_{yz|x}n_{zx}.
\end{align}

\section{Proof of Theorem~\ref{StrongLW-ddim}}\label{appendix:d}
We choose $(X_1, \ldots, X_d)$ uniformly at random from $S$. We then apply \eqref{eq:strongMTIneq} with $f(S) = H(X_S), S \subseteq [1:n]$, $\mathcal{F} = \{\!\!\{[1:d]\setminus{\{i\}}: i \in [1:d]\}\!\!\}$ and $\gamma(S) = 1/(d-1), ~\forall S \in \mathcal{F}$, to obtain an upper bound on $H(X_1, \ldots, X_d)$ as follows.
\begin{align}
    &H(X_1, \ldots, X_d) \nonumber \\
    & \leq \sum\limits_{S \in \mathcal{F}} \gamma(S)H(X_S|X_{< S}) \\
    & = \frac{1}{d-1} H(X_{[2:d]}|X_1) + \frac{1}{d-1} \sum\limits_{i = 2}^d H(X_{[1:d]\setminus{\{i\}}})\\
    & \leq \frac{1}{d-1} H(X_{[2:d]}|X_1) + \frac{1}{d-1} \sum\limits_{i = 2}^d \log{n_{-i}} \label{eq: SLWD_1}\\
    & \leq \frac{1}{d-1} \Big(\log{n^*_{-d|x_d}} + \sum\limits_{i = 1}^{d-1} \log{n_{-i}} \Big), \label{eq: SLWD_3}
\end{align}
where \eqref{eq: SLWD_1} follows from the definition of $n_{-i},~i \in [2:d]$, together with the fact that $H(W) \leq \log{|\mathcal{W}|}$, where $\mathcal{W}$ denotes the alphabet of $W$. \eqref{eq: SLWD_3} follows because
\begin{align}
    &H(X_{[2:d]}|X_1) \nonumber \\
    & = \sum\limits_{x} P_{X_1}(x)H(X_{[2:d]}|X_1 = x) \\
    & \leq \sum\limits_{x} P_{X_1}(x) \log{|\{(x_2, \ldots, x_d):(x,x_2, \ldots, x_d)\in S\}|} \label{eq: SLWD_4prime}\\
    & \leq \sum\limits_{x} P_{X_1}(x) \log{n^*_{-1|x_1}} \label{eq: SLWD_4}\\
    & = \log{n^*_{-1|x_1}},
\end{align}
where \eqref{eq: SLWD_4prime} holds since, conditioned on $X_1=x$, the alphabet of $(X_2,\dots,X_d)$ is given by 
\begin{align}\{(x_2, \ldots, x_d):(x,x_2, \ldots, x_d)\in S\},
\end{align}
 and \eqref{eq: SLWD_4} follows from the definition of $n^*_{-1|x_1}$. Combining \eqref{eq: SLWD_3} with the fact that $H(X_1,\ldots,X_d) = \log{n}$, we get
\begin{align}
    n^{d-1}\leq \left(\prod_{i=2}^{d}n_{-i}\right)n^*_{-1|x_1}.
\end{align}

\section{Proof of Theorem~\ref{ExtremalGraph}}\label{appendix:e}
Let $X^n = (X_1, \ldots, X_n)$ be chosen uniformly at random from $V(G)$. So, $P_{X^n}(x^n) = 1/|V(G)|$, for $x^n \in V(G)$ and $H(X^n) = \log{|V(G)|}$. For any $x^n \in V(G)$ and $S = \{i_1, \ldots, i_d\}$, we have either $\Bar{x}^{(S)} \in V(G)$ or $\Bar{x}^{(S)} \notin V(G)$. If $\Bar{x}^{(S)} \in V(G)$, then
\begin{align}
    P_{X_S|\Tilde{X}^{(S)}}&(x_S|\Tilde{x}^{(S)}) \nonumber \\
    & = \frac{1}{|\{y^n \in V(G): \Tilde{y}^{(S)} = \Tilde{x}^{(S)}, \Bar{x}^{(S)} \in V(G)\}|} \\
    & \leq \frac{1}{m_d}, \label{eq: EG_1}
\end{align}
where \eqref{eq: EG_1} holds by definition of $m_d$ in \eqref{eq: m_d}. Similarly, if $\Bar{x}^{(S)} \notin V(G)$, then
\begin{align}\label{eq: EG_2}
    P_{X_S|\Tilde{X}^{(S)}}&(x_S|\Tilde{x}^{(S)}) \leq \frac{1}{\ell_d}
\end{align}
by the definition of $l_d$ in \eqref{eq: l_d}.
Now, let us compute $H(X^n) - H(\Tilde{X}^{(S)})$.
\begin{align}
    & H(X^n) - H(\Tilde{X}^{(S)}) \nonumber \\
    & = H(X_S|\Tilde{X}^{(S)}) \label{eq: EG_3}\\
    & = - \sum\limits_{x^n \in \{-1,1\}^n} P_{X^n}(x^n) \log{\left(P_{X_S|\Tilde{X}^{(S)}}(x_S|\Tilde{x}^{(S)})\right)} \\
    & = - \frac{1}{|V(G)|} \sum\limits_{x^n \in V(G)} \log{\left(P_{X_S|\Tilde{X}^{(S)}}(x_S|\Tilde{x}^{(S)})\right)} \\
    & = - \frac{1}{|V(G)|} \sum\limits_{\substack{x^n \in V(G),\\ \Bar{x}^{(S)} \in V(G)}} \log{\left(P_{X_S|\Tilde{X}^{(S)}}(x_S|\Tilde{x}^{(S)})\right)} \nonumber \\
    & \quad\quad - \frac{1}{|V(G)|} \sum\limits_{\substack{x^n \in V(G),\\ \Bar{x}^{(S)} \notin V(G)}} \log{\left(P_{X_S|\Tilde{X}^{(S)}}(x_S|\Tilde{x}^{(S)})\right)} \\
    & \geq \frac{\log{m_d}}{|V(G)|} \sum\limits_{x^n \in \{-1,1\}^n} \mathbbm{1}\{x^n \in V(G),\Bar{x}^{(S)} \in V(G)\} \nonumber \\
    & \quad \quad + \frac{\log{\ell_d}}{|V(G)|} \sum\limits_{x^n \in \{-1,1\}^n} \mathbbm{1}\{x^n \in V(G),\Bar{x}^{(S)} \notin V(G)\}, \label{eq: EG_4}
\end{align}
where \eqref{eq: EG_3} follows by chain rule, and \eqref{eq: EG_4} follows from \eqref{eq: EG_1} and \eqref{eq: EG_2}. Summing both the sides of \eqref{eq: EG_4} over all sets $S \in \mathcal{F}_d$, we get
\begin{align}
    &\sum\limits_{S \in \mathcal{F}_d} \left( H(X^n) - H(\Tilde{X}^{(S)}) \right) \nonumber \\
    & \geq \frac{\log{m_d}}{|V(G)|} \sum\limits_{S \in \mathcal{F}_d} \sum\limits_{x^n \in \{-1,1\}^n} \mathbbm{1}\{x^n \in V(G),\Bar{x}^{(S)} \in V(G)\} \nonumber \\
    & \quad + \frac{\log{\ell_d}}{|V(G)|} \sum\limits_{S \in \mathcal{F}_d} \sum\limits_{x^n \in \{-1,1\}^n} \mathbbm{1}\{x^n \in V(G),\Bar{x}^{(S)} \notin V(G)\}. \label{eq: EG_5}
\end{align}
We can compute both the double summations on the RHS of \eqref{eq: EG_5} in terms of $|\mathcal{F}_d|,|V(G)|$, and $|E_{\mathcal{F}_d}(G)|$, where $E_{\mathcal{F}_d}(G)$ denotes the set of edges connecting pairs of vertices $x^n,y^n\in V(G)$ such that $\Delta(x^n,y^n)\in \mathcal{F}_d$. The first double summation is given by
\begin{align}
    \sum\limits_{S \in \mathcal{F}_d} \sum\limits_{x^n \in \{-1,1\}^n} \mathbbm{1}\{x^n \in V(G),\Bar{x}^{(S)} \in V(G)\} = 2|E_{\mathcal{F}_d}(G)|, \label{eq: EG_7}
\end{align}
because every edge $x^n,\Bar{x}^{(S)} \in E_{\mathcal{F}_d}(G)$ is counted twice in the above double summation. For the the second double summation, notice that
\begin{align}
    & \sum\limits_{S \in \mathcal{F}_d} \sum\limits_{x^n \in \{-1,1\}^n} \mathbbm{1}\{x^n \in V(G),\Bar{x}^{(S)} \in V(G)\} \nonumber \\
    & \quad \quad + \sum\limits_{S \in \mathcal{F}_d} \sum\limits_{x^n \in \{-1,1\}^n} \mathbbm{1}\{x^n \in V(G),\Bar{x}^{(S)} \notin V(G)\} \\
    & = \sum\limits_{S \in \mathcal{F}_d} \sum\limits_{x^n \in \{-1,1\}^n} \mathbbm{1}\{x^n \in V(G)\} \\
    & = |\mathcal{F}_d||V(G)|. \label{eq: EG_11}
\end{align}
Now, substituting \eqref{eq: EG_7} in \eqref{eq: EG_11} gives the value of the second double summation as
\begin{align}
    &\sum\limits_{S \in \mathcal{F}_d} \sum\limits_{x^n \in \{-1,1\}^n} \mathbbm{1}\{x^n \in V(G),\Bar{x}^{(S)} \notin V(G)\} \nonumber \\
    & \quad \quad = |\mathcal{F}_d||V(G)| - 2|E_{\mathcal{F}_d}(G)|. \label{eq: EG_8}
\end{align}
Substituting the values of the double summations from \eqref{eq: EG_7} and \eqref{eq: EG_8} in \eqref{eq: EG_5} gives
\begin{align}
    &\sum\limits_{S \in \mathcal{F}_d} \left( H(X^n) - H(\Tilde{X}^{(S)}) \right) \nonumber \\
    & \geq \frac{\log{m_d}}{|V(G)|}2|E_{\mathcal{F}_d}(G)| + \frac{\log{\ell_d}}{|V(G)|} \big(|\mathcal{F}_d||V(G)| - 2|E_{\mathcal{F}_d}(G)|\big). \label{eq: EG_9}
\end{align}
Also by Shearer's lemma, we have that
\begin{align}
    \sum\limits_{S \in \mathcal{F}_d} H(\Tilde{X}^{(S)}) &=\sum_{S^\text{c} \in \Bar{\mathcal{F}}_d}\!H(\Tilde{X}^{(S)})\\
    &\geq t_d H(X^n). \label{eq: EG_12}
\end{align}
Subtracting $|\mathcal{F}_d|H(X^n)$ on both the sides of \eqref{eq: EG_12} gives 
\begin{align}
    \sum\limits_{S \in \mathcal{F}_d} \left( H(X^n) - H(\Tilde{X}^{(S)}) \right) &\leq (|\mathcal{F}_d| - t_d)H(X^n) \\
    & = (|\mathcal{F}_d| - t_d) \log{|V(G)|}. \label{eq: EG_10}
\end{align}
By putting together \eqref{eq: EG_9} and \eqref{eq: EG_10}, we get
\begin{align}
    & \frac{\log{m_d}}{|V(G)|}2|E_{\mathcal{F}_d}(G)| + \frac{\log{\ell_d}}{|V(G)|} \big(|\mathcal{F}_d||V(G)| - 2|E_{\mathcal{F}_d}(G)|\big) \nonumber \\
    & \quad\quad \leq (|\mathcal{F}_d| - t_d) \log{|V(G)|}. \label{eq: EG_6}
\end{align}
On rearranging the terms in \eqref{eq: EG_6}, we get
\begin{align}
    |E_{\mathcal{F}_d}(G)| \leq \frac{|V(G)|\Big( (|\mathcal{F}_d| - t_d) \log{|V(G)|} - |\mathcal{F}_d| \log{\ell_d}\Big)}{2\log{\frac{m_d}{\ell_d}}}.
\end{align}

\section{Details of Remark~\ref{BLMSasonRemark}}\label{appendix:f}
To recover the bound obtained in \cite[Theorem~4.2]{BoucheronLM13}, we choose $\mathcal{F} = \{\!\!\{ \{i\}: i \in [1:n] \}\!\!\}$. For this family $\mathcal{F}$, there is only one possible value of $d$, namely, $d=1$. By definition, we have $m_1 = 2$ and $\ell_1 = 1$, and since $\bar{\mathcal{F}} = \{\!\!\{ \{i\}^{\text{c}}: i \in [1:n] \}\!\!\}$, it follows that $t_1 = n-1$. Substituting these values into \eqref{eq: EGBound}, we get $|E(G)| \leq \frac{1}{2}|V(G)|\log{|V(G)|}$.

To recover the bound obtained in \cite[Theorem~4.2]{BoucheronLM13}, we choose $\mathcal{F}=\{\!\!\{ S \subseteq [1:n]: |S| \leq \tau \}\!\!\}$. For each $d \in [1:\tau]$, we have $\mathcal{F}_d = \{\!\!\{ S \subseteq [1:n]: |S| = d \}\!\!\}$ and $\bar{\mathcal{F}}_d = \{\!\!\{ S \subseteq [1:n]: |S| = n-d \}\!\!\}$. Consequently, $t_d = \binom{n-1}{n-d-1} = \binom{n-1}{d}$. Substituting these values into \eqref{eq: EGBound}, we get,
\begin{align}
    &|E(G)| \nonumber \\
    &\leq \sum\limits_{d = 1}^\tau \frac{|V(G)|\Big( (|\mathcal{F}_d| - t_d) \log{|V(G)|} - |\mathcal{F}_d| \log{\ell_d}\Big)}{2\log{\frac{m_d}{\ell_d}}} \\
    & = \sum\limits_{d = 1}^\tau \frac{|V(G)|\Big( \big(\binom{n}{d} - \binom{n-1}{d} \big) \log{|V(G)|} - \binom{n}{d}\log{\ell_d}\Big)}{2\log{\frac{m_d}{\ell_d}}} \\
    & = \sum\limits_{d = 1}^\tau \frac{|V(G)|\Big( \binom{n-1}{d-1} \log{|V(G)|} - \binom{n}{d}\log{\ell_d}\Big)}{2\log{\frac{m_d}{\ell_d}}} \\
    & = \sum\limits_{d = 1}^\tau \frac{\binom{n-1}{d-1}|V(G)|\Big( \log{|V(G)|} - \frac{n}{d}\log{\ell_d}\Big)}{2\log{\frac{m_d}{\ell_d}}}.
\end{align}

\fi
\end{document}